\documentclass[Manuscript for review]{IEEEtran}
\usepackage{cite}
\usepackage{graphics}
\usepackage{graphicx}
\usepackage{epsfig}

\usepackage{amssymb}
\usepackage{amsthm}
\usepackage{amsmath}
\usepackage{enumerate}
\usepackage{color}
\usepackage{cases}
\usepackage{CJK}
\usepackage{indentfirst}
\usepackage{subeqnarray}

\newtheorem{assumption}{Assumption}
\newtheorem{lemma}{Lemma}
\newtheorem{definition}{Definition}
\newtheorem{theorem}{Theorem}
\newtheorem{remark}{Remark}

\begin{document}
\title{Fixed-Time Cooperative Tracking Control for Double-Integrator Multi-Agent Systems: A Time-Based Generator Approach}
%\title{Time-Varying Formation
%for General Linear
%Multi-Agent Systems Over Directed Topologies: A Fully Distributed Adaptive
%Technique}

%\author{
%Yu~Zhao, \emph{Member, IEEE}, Qixiu Duan,  Guanghui Wen, \emph{Senior Member, IEEE}, Dong Zhang,  \emph{Member, IEEE},  and  Bohui Wang, \emph{Member, IEEE}

\author{
Qiang Chen, Yu~Zhao, Guanghui Wen, Guoqing Shi and Xinghuo Yu
        % <-this % stops a space
%\thanks{This work is supported by the National Nature Science Foundation of China through Grant Nos. 61603301, 61603300, 61722303 and 61673104, the China Postdoctoral Science Foundation with Grant Nos. 2017M623244, 2018T111097, 2018M633575, the Natural Science Foundation of Shaanxi Province of China through Grant Nos. 2017JQ6016 and 2018JQ6073, 2018JQ6071, the Fundamental Research for Science and Technology Planning Project of Shenzhen through Grant No. JCYJ20170306153912850, the Source Innovation Program of Qindao through Grant No. 18-2-2-39-jch, the Natural Science Foundation of Jiangsu Province of China through Grant No. BK20170079, the Australian Research Council through Grant No. DE180101268, and the Fundamental Research Funds for the Central Universities of China through Grant Nos. 3102018jcc038, 2242018k1G004, 20103176478.}
%\thanks{Corresponding author: Yu Zhao}
\thanks{Q. Chen is with the School of Electronics and Information and the school of Automation, Northwestern Polytechnical
University, Xi'an 710129, China (e-mail: chenq@mail.nwpu.edu.cn).}% <-this % stops a space
%
%\thanks{Q. Chen is with the School of Electronics and Information, Northwestern Polytechnical
%University, Xi'an 710129, China
%(e-mail: chenq@mail.nwpu.edu.cn).}% <-this % stops a space
%
\thanks{Y. Zhao is with the School of Automation, Northwestern Polytechnical
University, Xi'an 710129, China
(Corresponding author).}
\thanks{G. Wen is with the School of Mathematics, Southeast University,
 Nanjing 210096, China.}
\thanks{G. Shi is with the School of Electronics and Information, Northwestern Polytechnical
University, Xi'an 710129, China.}
\thanks{X. Yu is with the School of Engineering, RMIT University, Melbourne, VIC 3001, Australia.}

%\thanks{D. Zhang is with the School of Aeronautics, Northwestern Polytechnical University,
%Xi'an 710129, China
%(e-mail: zhangdong@nwpu.edu.cn).}
%
%\thanks{B. Wang is with the School of Aerospace Science and Technology, Xidian University, Xi'an 710071, China
% (e-mail: wang31aa@126.com). }

% <-this % stops a space
%\thanks{W. Ren is with the Department of Electrical and Computer Engineering
%University of California, Riverside
%Riverside, CA 92521, USA (e-mail: ren@ee.ucr.edu). }
%
%
%\thanks{G. Chen is with the Department of Electronic Engineering, City University
%of Hong Kong, Kowloon, Hong Kong
%(e-mail: gchen@ee.cityu.edu.hk). }% <-this % stops a space
%\thanks{Copyright (c) 2014 IEEE. Personal use of this material is permitted.
%However, permission to use this material for any other purposes must be
%obtained from the IEEE by sending an email to pubs-permissions@ieee.org. Corresponding author: yuzhao5977@gmail.com, liuyongfangpku@gmail.com,  wenguanghui@gmail.com and gchen@ee.cityu.edu.hk. }
}


\markboth{
	IEEE TRANSACTIONS ON}%
{Shell \MakeLowercase{\textit{et al.}}:
	Fixed-Time Cooperative Tracking Control for Double-Integrator Multi-Agent Systems: A Time-Based Generator Approach}
%\markboth{
%IEEE TRANSACTIONS ON SYSTEMS, MAN, AND CYBERNETICS: SYSTEMS}%
%{Shell \MakeLowercase{\textit{et al.}}:
%Fixed-Time Distributed Average Tracking for Integrator-Type Multiagent Systems: A Time Base Generator Approach}



\maketitle

% \linespread{1.2}
\begin{abstract}
In this paper, both the fixed-time distributed consensus tracking and the fixed-time distributed average tracking problems for double-integrator-type multi-agent systems with bounded input disturbances are studied, respectively. Firstly, a new practical robust fixed-time sliding mode control method based on the time-based generator is proposed. Secondly, a fixed-time distributed consensus tracking observer for double-integrator-type multi-agent systems is designed to estimate the state disagreements between the leader and the followers under undirected and directed communication, respectively. Thirdly, a fixed-time distributed average tracking observer for double-integrator-type multi-agent systems is designed to measure the average value of reference signals under undirected communication. Note that both the observers for the distributed consensus tracking and the distributed average tracking are devised based on time-based generators and can be extended to that of high-order multi-agent systems trivially. Furthermore, by combing the fixed-time sliding mode control with the fixed-time observers, the fixed-time controllers are designed to solve the distributed consensus tracking and the distributed average tracking problems. Finally, a few numerical simulations are shown to verify the results.
\end{abstract}


% \fi
\IEEEpeerreviewmaketitle

\section{Introduction}
Distributed cooperation control has been a popular scientific research issue over the past decades owing to its significant value in reality such as distributed optimization \cite{Li2019}, \cite{Zhao2017}, tracking control \cite{Li2013a, Wen2018, Ren2019}, flocking and containment control \cite{Su2009Flocking, Wen1784, Wen2016}.

	In the distributed cooperation control of a flock of agents with local interactions, a premier task is to design an algorithm which makes each agent achieve consensus in position, velocity and so on. The consensus algorithm for single-integrator multi-agent systems was first developed in \cite{Olfati-Saber2004}, and then some sufficient and necessary conditions for the consensus of double-integrator multi-agent systems were generalized in \cite{YU20101089}. Distributed tracking control can be regarded as an extension of generalized consensus control, in which the followers not only have to reach consensus, but also to follow with the specified trajectory. For example, in the distributed consensus tracking and distributed average tracking, the target trajectories are the states of the leader and the average value of multiple reference signals, respectively. However, in distributed algorithms, only a few or none of the agents can acquire the target information directly. Therefore, a frequently-used method to measure the target information is to establish a distributed observer. An observer-based algorithm for nonlinear agents was proposed in \cite{Cao2017}, \cite{Yu2017} to achieve distributed consensus tracking. In \cite{Chen2015, ZHAO2017157, ZHAO2017158}, the distributed observers were designed to measure the average value of reference signals. However, these protocols are asymptotically stable, which implies that the upper-bounded convergence time is not guaranteed. So as to estimate the precise upper-bounded convergence time, the finite-time observers relying on initial conditions were proposed in \cite{CAO2010522, Zhao2013, Zhao2019c}. Unfortunately, in some engineering practices, the initial states are not available or the convergent rate has to be faster. Therefore developing the fast converging algorithms without dependence on initial states is quite necessary. The fixed-time stability strategy was first investigated in \cite{Polyakov2012}, in which the prerequisite of initial conditions was eliminated. Some novel fixed-time algorithms for single-integrator multi-agent systems were developed in \cite{Zuo2014A}. A fixed-time observer for double-integrator-type multi-agent systems to estimate the states of the leader was designed in \cite{FU20161} under undirected communication topology. Then A fixed-time directed edition for high-order integrator-type multi-agent systems was developed in \cite{Zuo2019a}; Although, after the system states converging into the unit circle in the last step it is asymptotically stable, this method provides some inspirations for the protocol in this paper.
% Nevertheless, in the process of extending the conventional fixed-time consensus algorithm of first-order systems to second-order systems, it is not as smoothly as %we think.

After estimating the task trajectory in a fixed-time, the next step is to devise the fixed-time controller for the agent to track the target trajectory. For double-integrator-type systems, sliding mode control is a type of classic nonlinear control protocol \cite{edwards1998sliding}, which has the advantages of fast response, parameter change, insensitive to disturbance and simple physical implementation. Some finite-time sliding mode control methods were proposed in \cite{WU1998281, FENG20131715}. However, the process of extending it to that of double-integrator-type systems is nontrivial due to the singular problem. An attempt to design the fixed-time sliding mode control protocol was made in \cite{ZUO2015305} by utilizing a sinusoid function to offset the singularity in the neighborhood of zero, but it leads to a little uncertainty of the convergence time. The price of reducing the uncertainty is a sharp rise of the control input.

Besides the conventional fixed-time protocol that usually use two feedback terms, another popular technique in terminal control is the time-based generator technique, which was induced in \cite{MORASSO1997411} to induce the attraction of force fields. In \cite{Parra-Vega2001}, an finite-time sliding mode surface is designed, but it can't accurately track the trajectory and has no robustness. In \cite{Becerra2018}, a predefined-time control method for single-input single-output controllable linear systems was proposed. A novel fixed-time consensus strategy for single-integrator multi-agent systems was developed in \cite{ning2019practical}. Furthermore, it proved that the fixed-time protocol based on the time-based generators had a less magnitude of control inputs.

 As for the time synchronization between different agents, the clock synchronization device has been proposed in the existing paper \cite{Bolognani2016, Carli2014} to ensure the time synchronization. Therefore, it is not repeated here.

	Motivated by the above results, by utilizing the time-based generator technique, five main contributions are made in this paper. Firstly, a new fixed-time nonsingular sliding mode control method is developed, which can precisely predesign the upper-bounded convergence time without dependence on initial states and has a less magnitude of control inputs compared with the conventional ones \cite{FU20161,ZUO2015305}. Secondly, a new fixed-time distributed observer under undirected topology is proposed to evaluate the state disagreements between the leader and the followers. Thirdly, inspired by \cite{Zuo2019a}, the observer for undirected communication systems is extended to the systems with directed communication, but what's different is that the observer in this paper is a fully fixed-time protocol with the precise upper-bounded convergence time. Note that, all the observers proposed in this paper can be extended to that of high-order multi-agent systems trivially. Moreover, by combining the sliding mode control protocol with the distributed consensus tracking observers, two fixed-time controllers are developed which successfully extended the fixed-time distributed consensus tracking algorithms based on time-based generators for single-integrator-type multi-agent systems in \cite{ning2019practical} to the double ones. More importantly, the disturbance is considered in this paper, which is of great significance in practice. Finally, a controller is given to solve the fixed-time distributed average tracking problems for double-integrator-type multi-agent systems. As far as I am concerned, there is no other fixed-time distributed average tracking algorithm for double-integrator-type multi-agent systems.

The rest of this paper is given as below. In section \ref{sec2}, some mathematical preliminaries were given. In section \ref{sec3}, the fixed-time sliding mode control protocol is investigated. Next, the observers for distributed consensus tracking under both undirected and directed graph are designed. Then the observer for distributed average tracking under undirected graph is designed. Furthermore, the distributed consensus tracking and the distributed average tracking algorithms are given. In section \ref{sec4}, several simulations are given. In section \ref{sec5}, a few conclusions are made.

\section{Mathematical preliminaries} \label{sec2}

\subsection{Notations}
The real number set and the N-dimensional real vector space are denoted by $\mathbb{R}$ and $\mathbb{R}^n$, respectively. The signum function is represented by ${\rm sgn}(\cdot)$ and its vector form can be written as ${\rm sgn}(z)=[{\rm sgn}(z_1), {\rm sgn}(z_2), ..., {\rm sgn}(z_n)]^T$, where $z=[z_1,z_2,...,z_n]^T$. Let $|\cdot|$ stand for the absolute value of a scalar. The vector q-norm can be written as $\|z\|_q=(|z_1|^q+|z_2|^q+...+|z_n|^q)^{\frac{1}{q}}$. Let $\lambda_1(Q)$ and $\lambda_2(Q)$ represent the smallest and the second smallest eigenvalues of the matrix $Q$, respectively.

\subsection{Graph Theory}
The communication topology of a group of $n+1$ agents can be represented by a graph $\mathcal{G}$. If there is a leader in them, the other $n$ agents can be expressed as a subgraph $\mathcal{G}_s$. The weighted graph $\mathcal{G}=(\mathcal{\mathcal{V}},\mathcal{E})$ is constructed with a set of nodes $\mathcal{V}=\left\lbrace v_1,v_2,...,v_{n+1}\right\rbrace$ and a set of edges $\mathcal{E}=\left\lbrace e_1,e_2,...,e_m\right\rbrace$. A directed edge from $v_j$ to $v_i$ can be denoted as $(v_i,v_j)$, which means $v_i$ can receive information from $v_j$. An directed path from $v_j$ to $v_i$ consists of a sequence of edges in the form of $\mathcal{E}_{ij}=\left\lbrace (v_i,v_{i+1}),...,(v_{j-1},v_j)\right\rbrace$, which means the information can flow from $v_j$ to $v_i$. When replacing the directed edges by the undirected, it becomes undirected path and the information flow is bidirectional. It is said to contain a spanning tree if at leat there exists a node which has directed paths to all other nodes. The undirected graph is connected if and only if there at least exists an undirected path between any two notes. Let $A=[a_{ij}]\in\mathbb{R}^{n \times n}$ and $D\in\mathbb{R}^{n \times m}$ denote the adjacency matrix and the incidence matrix of the graph respectively, and $a_{ij}=1$ if there exists a directed edge from $v_j$ to $v_i$, else $a_{ij}=0$. With regard to undirected graphs, $a_{ij}=a_{ji}$. let $O=[o_{ij}]\in \mathbb{R}^{n \times n}$ denote the degree matrix and $o_{ii}=\sum_{j=1}^{n}a_{ij}$, else $p_{ij}=0$. Then the Laplacian matrix is written as $L\in\mathbb{R}^{n \times n}=O-A$. Set $a_{0i}=0$ and $a_{i0}=1$ if the agent $i$ can acquire information from the leader, else $a_{i0}=0$ and then set $B={\rm diag}\{a_{i0},a_{i1},...,a_{in}\}$.

\subsection{Time-Based Generator}
The time-based generator $\xi(t)$ is a kind of time dependent function that can be seen as a termination function. Its general properties can be generalized as follows.
\begin{enumerate}
\item $\xi(t)$ is a non-decreasing and continuous function.
\item With time going by, $\xi(t)$ increases from the initial state $\xi(0)=0$ to $\xi(t_s)=1$, and when $t>t_{s}$, $\xi(t)\equiv 1$, where $t_s$ can be predesigned arbitrarily.
\item ${\dot\xi}(0)=0$ and when $t\geq t_{s}$, ${\dot\xi}(t)\equiv 0$.
\end{enumerate}
\begin{remark}
A typical time-based generator function $\xi(t)$ is presented as follows \cite{ning2019practical}.
\begin{eqnarray}
\begin{cases}{\xi(t)=}
\frac{10}{t_s^6}t^6-\frac{24}{t_s^5}t^5+\frac{15}{t_s^4}t^4,\;0\leq t\leq t_s, \nonumber\\
1, \qquad\qquad\qquad\qquad\qquad\quad\;\; t>t_s,
\end{cases}
\end{eqnarray}

\end{remark}
Think about the differential equation as below. %%%%%%%%%% TBG 方程
\begin{align}\label{b1}
\dot{z}=-h(t)z, z(0)=z_{0},
\end{align}
where $h(t)$ is constructed as
\begin{align}\label{b2}
h(t)=k\frac{\dot{\xi}}{1-\xi+\delta},
\end{align}
where $k\in\mathbb{R}$ and $\delta\in\mathbb{R}$ are two positive constants which satisfy $k>1$ and $0<\delta<<1$.

Solving the differential equation (\ref{b1}) one has
\begin{align}\label{b3}
  z=(\frac{1-\xi+\delta}{1+\delta})^k z_0.
\end{align}
With $t$ growing from $0$ to $t_s$, $\xi$ grows from $0$ to $1$ smoothly. Therefore, when $t\in[0,t_s)$, $z$ gradually approaches $z_0(\frac{\delta}{1+\delta})^k$. When $t\geq t_s$, the result will remain the same. If let $\delta=0.001$ and $k=3$, at the terminal moment $t_s$, the solution of (\ref{b1}) will be $z=10^{-9}z_0$. Thus we can nearly think that $z$ reaches zero at $t_s$ and the initial state $z_0$ has no effect on the convergence time.

%\subsection{Useful lemmas}            %%%%%%%%%%%%%%%%% 引理

%\begin{lemma} \label{lemma1}
%If $Q\in\mathbb{R}^{n\times n}$ is a positive definite matrix and $\lambda_1(Q)$ is its minimum eigenvalue, for any vector $z\in\mathbb{R}^n$ one has
%\begin{align}\label{b4}
%z^TQz\geq \lambda_1(Q)z^Tz.
%\end{align}
%\end{lemma}

%\begin{lemma}\label{lemma3}
%
%\end{lemma}

\subsection{Problem Description}
\subsubsection{Fixed-Time Distributed Consensus Tracking}
Suppose that there is a double-integrator-type multi-agent system with a leader and n agents. The leader can be represented by
\begin{align}\label{b7}                   %%%%%%%%%%%% 领导者
\begin{cases}
  \dot{x}_0=v_0,\\
  \dot{v}_0=u_0,
\end{cases}
\end{align}
where $x_0\in\mathbb{R}$ and $v_0\in\mathbb{R}$ represent the position and velocity of the leader, respectively. $u_0\in\mathbb{R}$ represents the control input bounded by a positive constant $u_{max}$.

Then the followers can be modeled by           %%%%%%%%%%跟随者
\begin{align}\label{b8}
\begin{cases}
\dot{x}_i=v_i,\\
\dot{v}_i=u_i+d_i,\; i=1,2,...,n,
\end{cases}
\end{align}
where $x_i\in\mathbb{R}$ and $v_i\in\mathbb{R}$ denote the position and the velocity of the agent $i$, respectively. $u_i\in\mathbb{R}$ and $d_i\in\mathbb{R}$ denote the control input and the uncertainty, where $d_i$ takes the positive constant $d_{max}$ as the boundary.

The objective of fixed-time distributed consensus tracking is to devise the control input only using local information for each follower, which enable the followers to achieve consensus with the leader in a fixed time independent of initial states.
{\begin{definition}\label{d1}(Fixed-time distributed consensus tracking)
For the system described by $(\ref{b7})$ and $(\ref{b8})$, with the given observer and control input $u_i$, it is said to achieve fixed-time distributed consensus tracking if all the followers can achieve consensus with the leader in a fixed-time $T_{max}$ independent of initial conditions, i.e.,
\begin{align}\label{b9}
\begin{cases}
   \lim_{t\rightarrow T_{max}}|x_i-x_0|+|v_i-v_0|\leq c \\
   \lim_{t\rightarrow \infty}|x_i-x_0|+|v_i-v_0|=0,
\end{cases}
\end{align}
%%%%%%%%%%%%%%%%%%%%%%%%%%%%%%%%%%%  equation 3
where $T_{max}$ can be predesignated arbitrarily independent of initial conditions and $c$ can be limited to the desired level.
\end{definition}}
\subsubsection{Fixed-Time Distributed Average Tracking}
Consider a double-integrator-type multi-agent system with n agents represented by $(\ref{b8})$, and each agent $i$ has a reference signal $r_i\in\mathbb{R}$ described as follows.
\begin{align}\label{bb9}%%%%%%%%%%%%%%%%%%%   平均跟踪
\begin{cases}
\dot{r}_i=f_i,\\
\dot{f}_i=a^r_i,\; i=1,2,...,n,
\end{cases}
\end{align}
where $f_i$ and $a^r_i$ are the velocity and acceleration of reference signal $r_i$, respectively. Note that $a^r_i$ is bounded by a positive constant $a_{max}$. Let $\overline{r}=\frac{1}{n}\sum_{i=1}^{n}r_i$, $\overline{f}=\frac{1}{n}\sum_{i=1}^{n}f_i$ and $\overline{a}=\frac{1}{n}\sum_{i=1}^{n}f_i$ be the average value of the reference signals.

The objective of fixed-time distributed average tracking is to devise control inputs only using local information for the agents, which enable them to achieve consensus with the average value of multiple reference signals in a fixed time without dependence on initial states.
\begin{definition}\label{dd2}(Fixed-time distributed average tracking)
For the system described by $(\ref{b8})$ and $(\ref{bb9})$, with the given observer and control input $u_i$, it is said to achieve fixed-time distributed average tracking if all the agents can achieve consensus with the average value of the multiple reference signals in a fixed-time $T_{max}$ which can be predesigned arbitrarily and independent of initial states, i.e.,
\begin{align}\label{bb10}
\begin{cases}
   \lim_{t\rightarrow T_{max}}|x_i-\overline{r}|+|v_i-\overline{f}|\leq c \\
   \lim_{t\rightarrow \infty}|x_i-\overline{r}|+|v_i-\overline{r}|=0,
\end{cases}
\end{align}
\end{definition}
%%%%%%%%%%%%%%%%%%%%%%%%%%%%%%%%%%  equation 3
%\begin{remark}
%To clarify the idea, all the system states in this paper are assumed to be one-dimensional.
%\end{remark}
\section{Main Results}\label{sec3}
\subsection{Fixed-Time Sliding Mode Control}
\begin{lemma}\label{lemma2}\cite{Zuo2019a}
Suppose that $z(0)=z_0$ and $V(z)$ is a positive definite Lyapunov candidate which satisfies the inequality as below.
\begin{align}\label{b5}
  \dot{V}(z)+\mu V^\nu(z)\leq0,
\end{align}
where $\mu\geq 0$ and $\nu\in(0,1)$. Then $z$ will converge to zero in a finite time $T(z_0)$ such that
\begin{align}\label{b6}
  T(z_0)\leq\frac{1}{\mu(1-\nu)}V^{1-\nu}(z_0).
\end{align}
\end{lemma}
A typical double-integrator-type control system is given as follows.
\begin{align}\label{b10}
\begin{cases}
  \dot{z_1}=z_2,\\
  \dot{z_2}=u+\varrho,
\end{cases}
\end{align}
where $z_1\in\mathbb{R}$ and $z_2\in\mathbb{R}$ are the system states. $\delta\in\mathbb{R}$ is a disturbance bounded by a positive constant $\varrho_{max}$.

 The objective of fixed-time sliding mode control is to devise a control input $u$ which drives the system (\ref{b10}) to the equilibrium point in a fixed time, i.e., $[z_1,z_2]=[0,0]$. The process of fixed-time double-integrator sliding mode control is generally divided into two sections. In the first section, the control input forces the system to arrive at the prescribed surface in a fixed time $t_{a1}$; In the second section, the system will slide along the surface to the equilibrium point in a fixed time $t_{a2}$. Therefore, The whole convergence time is bounded by $T_a=t_{a1}+t_{a2}$. In order to converge in the fixed time in each stage, two time-based generators $\xi_{a1}$ and $\xi_{a2}$ are used sequentially. $\xi_{a1}$ ensures the system to arrive at the prescribed surface in $t_{a1}$ and then invalid. $\xi_{a2}$ guarantees the fixed convergence time $t_{a2}$. Let $h_{a1}(t)=k\frac{\dot{\xi}_{a1}}{1-\xi_{a1}+\delta}$ and $h_{a2}(t)=k\frac{\dot{\xi}_{a2}}{1-\xi_{a2}+\delta}$. Then we have
\begin{align}\label{b11}
h_1(t)=
\begin{cases}
  h_{a1}(t),t\in[0,t_{a1}), \\
  h_{a2}(t),t\in[t_{a1},t_{a1}+t_{a2}),\\
  0,t\in[t_{a1}+t_{a2},+\infty).
\end{cases}
\end{align}
\begin{remark}
  Since $\dot{\xi}_{a1}(0)=\dot{\xi}_{a1}(t_{a1})=\dot{\xi}_{a2}(t_{a1})=\dot{\xi}_{a2}(t_{a1}+t_{a2})=0$, one obtains $h_{a1}(0)=h_{a1}(t_{a1})=h_{a2}(t_{a1})=h_{a2}(t_{a1}+t_{a2})=0$, which shows the connectivity of $h_1(t)$. Furthermore, owing to the nonnegativity of $\dot{\xi}_{a1}$ and $\dot{\xi}_{a2}$, $h_1(t)$ is also nonnegative.
\end{remark}
\begin{remark}
  In order to clarify the idea, let $\xi_{a1}$, $h_{a1}(t)$ and $\xi_{a2}$, $h_{a2}(t)$ have the same structure. But, in simulation, owing to $\xi_{a1}(t_{a1})=1$ and $\xi_{a2}(t_{a1})=0$, which leads to a sharp decrease of the derivative and causes problems. By resetting $\hat{\xi}_{a2}(t)=\xi_{a2}+1$ and $\hat{h}_{a2}(t)=k\frac{\dot{\xi}_{a2}}{2-\xi_{a2}+\delta}$, the problem caused by discontinuity is solved. Moreover, in the different steps, $k$ and $\delta$ can be selected as different constants respectively.
\end{remark}

In this paper, the fixed-time sliding mode surface is selected as
\begin{align}\label{b12}
 s=(\frac{1}{2} h_1(t)+1)z_1+z_2.
\end{align}
If $s=0$, the system arrives at the sliding mode surface and has the form as below.
\begin{align}\label{b13}
  z_2=\dot{z}_1=-(\frac{1}{2}h_1(t)+1)z_1.
\end{align}
The control input is devised as follows.
\begin{align}\label{b14}
  u=&-\frac{1}{2}\dot{h}_1(t)z_1-(\frac{1}{2}h_1(t)+1)z_2\nonumber\\
    &-\frac{1}{2}h_1(t)s-\rho{\rm sgn}(s),
\end{align}
where $\rho$ is a positive constant satisfying $\rho \geq |\varrho_{max}|+1$.
\begin{theorem}\label{t1}
  With the given control input (\ref{b14}), the system (\ref{b10}) will arrive at the sliding mode surface ($s=0$) in a fixed-time $t_{a1}$, and then slide along the surface ($s=0$) to the equilibrium point $[z_1,z_2]=[0,0]$ in a fixed-time $t_{a2}$. Thus the final upper-bounded convergence time is $T_a=t_{a1}+t_{a2}$.
\end{theorem}
\begin{proof}
The Lyapunov candidate is constructed as $V_1=\frac{1}{2}s^2$. Differentiating (\ref{b12}) against time one has
\begin{align}\label{b15}
  \dot{s}=\frac{1}{2}\dot{h}_1(t)z_1+(\frac{1}{2}h_1(t)+1)z_2+u+\varrho.
\end{align}
Substituting the control input (\ref{b14}) in to (\ref{b15}) one has
\begin{align}\label{b16}
  \dot{s}=-\frac{1}{2}h_1(t)s-\rho{\rm sgn}(s)+\varrho.
\end{align}
Differentiating the Lyapunov candidate $V_1$ against time and then substituting (\ref{b16}) into it one has
\begin{align}\label{b17}
  \dot{V_1}&=s\dot{s}\nonumber\\
  &=-\frac{1}{2}h_1(t)s^2-\rho|s|+\varrho s \nonumber\\
  &\leq -\frac{1}{2}h_1(t)s^2-(\rho-\varrho_{max})|s| \nonumber\\
  &\leq -\frac{1}{2}h_1(t)s^2 \nonumber\\
  &=-h_1(t)V_1.
\end{align}

When $t\in[0,t_{a1}), h_1(t)=h_{a1}(t)$. According to (\ref{b1}) one obtains
%\begin{align}\label{b18}
%  \dot{V}_1=-h_{a1}(t)V_1,
%\end{align}
\begin{align}\label{b19}
  \lim_{t\rightarrow t_{a1}}V_1\leq(\frac{1-\xi_{a1}+\delta}{1+\delta})^k V_1(0)=(\frac{\delta}{1+\delta})^k V_1(0),
\end{align}
where, according to (\ref{b3}), $(\frac{\delta}{1+\delta})^kV_1(0)$ is in the near region of zero.

When $t\geq t_{a1}, h_1(t)=h_{a2}(t)$ and one obtains
\begin{align}\label{b20}
   \dot{V}_1=&-\frac{1}{2}h_{a2}(t)s^2-\rho|s|+\delta s\nonumber\\
   \leq&-(\rho-\delta_{max})|s|\nonumber\\
   \leq&-|s|\nonumber\\
   =&-\sqrt{2V_1}
\end{align}
According to Lemma \ref{lemma2} one has that $V_1$ will converge to zero after $t_{a1}$ within a finite time $\hat{t}_{a1}$, i.e.,  $\hat{t}_{a1}\leq\sqrt{2V_1(t_{a1})}\leq\sqrt{2V_1(0)(\frac{\delta}{1+\delta})^k}$. Although $V_1$ doesn't converge to zero perfectly at $t_{a1}$, which means that the system states are in the near region of the sliding mode surface ($s=0$), the system states will still converge to zero along the sliding surface in the fixed-time. In order to clarify the idea clearly, at first, it is assumed that $V_1$ converges to zero at $t_{a1}$, which means $s=0$ as well. Then the case that there is a small error between the sliding surface and the system states at $t_{a1}$ is investigated.

 A Lyapunov candidate is constructed as $V_2=\frac{1}{2}z_1^2$.
% \begin{align}\label{b21}
%   \dot{z}_1=-(\frac{1}{2} h_1(t)+1)z_1.
% \end{align}
 Differentiating $V_2$ along (\ref{b13}) one has
 \begin{align}\label{b22}
   \dot{V}_2&=z_1\dot{z_1} \nonumber\\
   &=-\frac{1}{2} h_1(t)z_1^2-z_1^2 \nonumber\\
   &\leq-\frac{1}{2} h_1(t)z_1^2\nonumber\\
   &=- h_1(t)V_2
 \end{align}
 From (\ref{b1}) one obtains that $\lim_{t\rightarrow {t_{a1}+t_{a2}}}V_2=(\frac{\delta}{1+\delta})^kV_2(t_{a1})$. When $t\geq t_{a1}+t_{a2}$, owing to $ h_1(t)=0$ and $\dot{V}_2=-z_1^2=-2V_2$ one concludes that $V_2$ will converge to zero exponentially. Since $z_2=-(h_1(t)+1)z_1=-z_1$, $z_2$ will converge to zero with the same rate of $z_1$ as well.

  In the following proof, the influence of $V_1(t_{a1})\neq 0$ is analysed. According to the relationship between $V_1$ and $s$, suppose that the system states converge to the adjacent region of the sliding surface at $t_{a1}$ and there exists an error $e$, i.e., $s=\{e|t\geq t_{a1}\}$. When $t\geq t_{a1}$, $\dot{V}_1\leq -|s|$, which means $V_1$ as well as $|e|$ are non-increasing functions and bounded by $\hat{V}_1=(\frac{\delta}{1+\delta})^k V_1(0)$ and $\hat{|e|}=\sqrt{2(\frac{\delta}{1+\delta})^k V_1(0)}$, respectively. Then (\ref{b12}) can be rewritten as
 \begin{align}\label{b23}
   e=(\frac{1}{2} h_1(t)+1)z_1+z_2.
 \end{align}
 The derivative of $z_1$ can be obtained as
 \begin{align}\label{b24}
   \dot{z_1}=-(\frac{1}{2} h_1(t)+1)z_1+e.
 \end{align}
 Substituting (\ref{b24}) into $\dot{V}_2$ one has
 \begin{align}\label{b25}
   \dot{V}_2&=-\frac{1}{2} h_1(t)z_1^2-z_1^2+ez_1 \nonumber\\
    &=-h_1(t)V_2-z_1^2+ez_1.
 \end{align}
 Note that $|e|$ is very small and non-increasing, and the convergence time is bounded by $\hat{t}_{a1}$. If $|z_1|< |e|$, which means that $|z_1|$ has been in the near region of zero. Meanwhile, $|z_1|$ is bounded by $\hat{|e|}$. After $t_{a1}+\hat{t}_{a1}$, due to $|e|=0$, $|z_1|$ will at least converge exponentially.
 If $|z_1|\geq |e|$, one has $\dot{V}_2\leq -h_1(t)V_2$. From (\ref{b1}), $V_2$ will nearly converge to zero within a fixed time $t_{a2}$ or converge to $|\hat{e}|$ in $\hat{t}_{a1}$, Whatever the case may be, $V_2$ will converge nearly to zero in a fixed time $T_a=t_{a1}+t_{a2}$. Then the proof has been completed.
\end{proof}
\subsection{Fixed-Time Distributed Consensus Tracking Observer Under Undirected Communication}
{\begin{assumption}\label{assumption1}
  The topology subgraph $\mathcal{G}_s$ for the followers is undirected and connected; There at least exists a follower which can acquire information from the leader.
\end{assumption}}
\begin{lemma}\label{lemma0}\cite{FU20161}
  If $L\in\mathbb{R}^{n\times n}$ is the Laplacian matrix of a undirected connected graph, and the nonnegative diagonal matrix $B={\rm diag}\{a_{10},...,a_{n0}\}$ with at least one element greater than zero, then $Q=L+B$ is a positive definite matrix.
  \end{lemma}
In the subsection, a fixed-time distributed observer based on time-based generators is designed for each follower to measure the relative position and velocity disagreements between the leader and itself under undirected communication. set two time-based generators as $\xi_{b1}$ and $\xi_{b2}$, and then $h_{b1}(t)=\frac{k\dot{\xi}_{b1}}{1-\xi_{b1}+\delta}$ and $h_{b2}(t)=\frac{k\dot{\xi}_{b2}}{1-\xi_{b2}+\delta}$. Let $h_2(t)$ has the same structure as $h_1(t)$. Denote the real tracking errors as $\tilde{x}_i=x_i-x_0$ and $\tilde{v_i}=v_i-v_0$. Then fixed-time distributed the observers $\alpha_i$ and $\beta_i$ of estimating $\tilde{x}_i$ and $\tilde{v}_i$ are proposed as below.
\begin{align}\label{b26}
  \dot{\alpha}_i=&\beta_i-b_1h_2(t)\bigg\{\sum_{j=0}^{n}a_{ij}[(\alpha_i-\alpha_j)-(x_i-x_j)]\bigg\} \nonumber\\
  &-b_2{\rm sgn}\bigg\{\sum_{j=0}^{n}a_{ij}[(\alpha_i-\alpha_j)-(x_i-x_j)]\bigg\},\nonumber\\
  \dot{\beta}_i=&u_i-c_1h_2(t)\bigg\{\sum_{j=0}^{n}a_{ij}[(\beta_i-\beta_j)-(v_i-v_j)]\bigg\} \nonumber\\
  &-c_2{\rm sgn}\bigg\{\sum_{j=0}^{n}a_{ij}[(\beta_i-\beta_j)-(v_i-v_j)]\bigg\},
\end{align}
where $i=1,...,n$, $\alpha_0=0$, $\beta_0=0$. $b_1$, $b_2$, $c_1$ and $c_2$ are positive constants satisfying $b_1=c_1\geq\frac{1}{2\lambda_1(Q)}$, $b_2\geq 1$ and $c_2> u_{max}+d_{max}$.

Let $\tilde{\alpha}_i=\alpha_i-\tilde{x}_i$ and $\tilde{\beta}_i=\beta_i-\tilde{v}_i$ be the errors between the observing disagreements and the real disagreements. If all the errors converge to zero in $t_{b1}+t_{b2}$, the observer is designed successfully.
\begin{theorem}
  With the given dynamics (\ref{b7}), (\ref{b8}) and observer (\ref{b26}), under Assumption \ref{assumption1}, $\alpha_i$ and $\beta_i$ converges to $\tilde{x}_i$ and $\tilde{v}_i$ within a fixed-time $T_b=t_{b1}+t_{b_2}$.
\end{theorem}
\begin{proof}
Following from (\ref{b26}), $\dot{\tilde{\alpha}}_i$ and $\dot{\tilde{\beta}}_i$ can be written as
\begin{align}\label{b27}
     \dot{\tilde{\alpha}}_i=&\tilde{\beta}_i-b_1h_2(t)\sum_{j=0}^{n}a_{ij}(\tilde{\alpha}_i-\tilde{\alpha}_j)\nonumber\\
    &-b_2{\rm sgn}\bigg[\sum_{j=0}^{n}a_{ij}(\tilde{\alpha}_i-\tilde{\alpha}_j)\bigg],\nonumber\\
     \dot{\tilde{\beta}}_i=&-c_1h_2(t)\sum_{j=0}^{n}a_{ij}(\tilde{\beta}_i-\tilde{\beta}_j)\nonumber\\
    &-c_2{\rm sgn}\bigg[\sum_{j=0}^{n}a_{ij}(\tilde{\beta}_i-\tilde{\beta}_j)\bigg]-d_i+u_0.
\end{align}
Let $\tilde{\alpha}=[\tilde{\alpha}_1,...,\tilde{\alpha}_n]^T$, $\tilde{\beta}=[\tilde{\beta}_1,...,\tilde{\beta}_n]^T$, $d=[d_1,...,d_n]^T$ and $u\in\mathbb{R}^n=[u_0,...,u_0]^T$. The vector form of (\ref{b27}) can be written as
\begin{align}\label{b28}
 \dot{\tilde{\alpha}}&=\tilde{\beta}-b_1h_2(t)Q\tilde{\alpha}-b_2{\rm sgn}(Q\tilde{\alpha}),\nonumber\\
 \dot{\tilde{\beta}}&=-c_1h_2(t)Q\tilde{\beta}-c_2{\rm sgn}(Q\tilde{\beta})-d+u,
\end{align}
where $Q=L+B$ is a positive definite matrix according to Lemma \ref{lemma0}.

Construct a Lyapunov candidate as $V_3=\frac{1}{2}\tilde{\beta}^TQ\tilde{\beta}$. Because $Q$ is a positive definite matrix, $V_3$ is well defined. Differentiate $V_3$ against time such that
\begin{align}\label{b29}
  \dot{V}_3&=\tilde{\beta}^TQ[-c_1h_2(t)Q\tilde{\beta}-c_2{\rm sgn}(Q\tilde{\beta})-d+u]\nonumber\\
  &\leq-c_1h_2(t)(Q^{\frac{1}{2}}\tilde{\beta})^TQ(Q^{\frac{1}{2}}\tilde{\beta})-(c_2-u_{max}-d_{max})||Q\tilde{\beta}||_1\nonumber\\
  &\leq-c_1\lambda_1(Q)h_2(t)\tilde{\beta}^TQ\tilde{\beta}\nonumber\\
  &\leq-h_2(t)V_3.
\end{align}

When $t\in[0,t_{b1})$, $h_2(t)=h_{b1}(t)$. According to (\ref{b1}) one concludes that $\lim_{t\rightarrow t_{b1}}V_3\leq(\frac{\delta}{1+\delta})^kV_3(0)<<1$. Then due to
\begin{align}\label{b29b}
  V_3&=\frac{1}{2}\tilde{\beta}^T(L+B)\tilde{\beta}\nonumber\\
  &=\frac{1}{4}\sum_{i=1}^{n}\sum_{j=1}^{n}a_{ij}(\tilde{\beta}_i-\tilde{\beta}_j)^2+\frac{1}{2}\sum_{i=1}^{n}a_{i0}\tilde{\beta}_i^2,
\end{align}
one has that $\lim_{t\rightarrow t_{b1}}|\tilde{\beta}_i|\leq 2\sqrt{(\frac{\delta}{1+\delta})^kV_3(0)}<<1$.

 When $t\geq t_{b1}$, one has
\begin{align}\label{b30}
\dot{V}_3=-(\theta_2-u_{max}-d_{max})||Q\tilde{\beta}||_1.
\end{align}
Owing to
\begin{align}\label{b31}
||Q\tilde{\beta}||_1\geq||Q\tilde{\beta}||_2=\sqrt{(Q\tilde{\beta})^TQ\tilde{\beta}}\geq\sqrt{\lambda_1(Q)\tilde{\beta}^TQ\tilde{\beta}},
\end{align}
one obtains
\begin{align}\label{b32}
  \dot{V_3}\leq-(\theta_2-u_{max}-d_{max})\sqrt{2\lambda_1(Q)V_3}\leq 0.
\end{align}
Therefore, $V_3$ will keep decreasing and $|\tilde{\beta}_i|<<1$ is ensured. Following from Lemma \ref{lemma2}, when $t\geq t_{b1}$, $V_3$ will converge to zero within a finite time $\hat{t}_{b1}$, i.e., $\hat{t}_{b1}\leq (\theta_2-u_{max}-d_{max})\sqrt{(\frac{\delta}{1+\delta})^k \frac{2V_3(0)}{\lambda_1(Q)}}$.

Construct a Lyapunov candidate as $V_4=\frac{1}{2}\tilde{\alpha}^TQ\tilde{\alpha}$. Differentiate it against time and then one has
\begin{align}\label{b33}
  \dot{V_4}&=\tilde{\alpha}^TQ\dot{\tilde{\alpha}} \nonumber\\
  &=\tilde{\alpha}^TQ\tilde{\beta}-b_1h_2(t)\tilde{\alpha}^TQQ\tilde{\alpha}-b_2\tilde{\alpha}^TQ{\rm sgn}(Q\tilde{\alpha})\nonumber\\
  &=-b_1h_2(t)\tilde{\alpha}^TQQ\tilde{\alpha}-b_2||Q\tilde{\alpha}||_1+(Q\tilde{\alpha})^T\tilde{\beta}.
\end{align}
Since when $t\geq t_{b1}$, $|\tilde{\beta}_i|<<1$, one has
\begin{align}\label{b34}
  \dot{V_4}\leq&-b_1\lambda_1(Q)h_2(t)\tilde{\alpha}^TQ\tilde{\alpha}-(b_2-1)||Q\tilde{\alpha}||_1\nonumber\\
  \leq&-h_2(t)V_4.
\end{align}
When $t\in[t_{b1},t_{b1}+t_{b2})$, $h_2(t)=h_{b2}(t)$. According to (\ref{b1}), one concludes that $\lim_{t\rightarrow t_{b1}+t_{b2}}V_4\leq(\frac{\delta}{1+\delta})^kV_4(t_{b1})\approx 0$.
When $t\geq t_{b1}+t_{b2}$, due to $h_2(t)=0$ and $\tilde{\beta}=0$, from (\ref{b33}) one has $\dot{V}_4=-b_2||Q\tilde{\alpha}||_1$. Compared with (\ref{b32}), one concludes that $V_4$ will converge to zero within finite time after $t_{b1}+t_{b2}$. That also means that the observer successfully complete the observing task in a fixed time $T_b=t_{b1}+t_{b2}$. Thus the whole proof has been finished.
\end{proof}
%%%%%%%%%%%%%%%%%%%%%%%%%%%%%%%%%%%%%%%%%%%%%%%%%%%%%%%%%%%%%%%%%%%%%%%%%%%%%%%%%%%%%%%%%%%%%%%%%%%%%%%%%%%%%%%%%%%%%%%%%%%%%%%%%%%%%%%%%%%
\subsection{Fixed-Time Distributed Consensus Tracking Observer Under Directed Communication}
\begin{assumption}\label{tree}
There is a spanning tree in the directed graph $\mathcal{G}$, where the leader is set as the root node. Note that, the subgraph $\mathcal{G}_s$ doesn't need to be strongly connected or contain a spanning tree.
\end{assumption}
\begin{lemma}\cite{Li2015}
  Under Assumption \ref{tree}, define $H=L+B$, $p=[p_1,...,p_n]^T=H^{-T}1_n$, $P={\rm diag}\{p_i\}$, $Q=\frac{H^TP+PH}{2}$. Then we have that $P$ and $Q$ are both positive definite.
\end{lemma}
Let $\overline{d}=|\mathop{{\rm max}}\limits_{i} \{\sum_{j=0}^{n}a_{ij}(d_i-d_j)\}|$ and $p_{max}=\mathop{{\rm max}}\limits_{i}\{p_i\}$, where $d_0=0$. Then the observer is given as below.

\begin{align}\label{b35}
\dot{\alpha}_i=&\beta_i-2b_1[h_2(t)+2]\bigg\{\sum_{j=0}^{n}a_{ij}[(\alpha_i-\alpha_j)-(x_i-x_j)]\bigg\}\nonumber\\
&-b_2[h_2(t)+2]{\rm sgn}\bigg\{\sum_{j=0}^{n}a_{ij}[(\alpha_i-\alpha_j)-(x_i-x_j)]\bigg\},\nonumber\\
\dot{\beta}_i=&u_i-2c_1[h_2(t)+2]\bigg\{\sum_{j=0}^{n}a_{ij}[(\beta_i-\beta_j)-(v_i-v_j)]\bigg\}\nonumber\\
&-c_2[h_2(t)+2]{\rm sgn}\bigg\{\sum_{j=0}^{n}a_{ij}[(\beta_i-\beta_j)-(v_i-v_j)]\bigg\},
\end{align}
where $\alpha_0=\beta_0=0$, and $b_1$, $b_2$, $c_1$, $c_2$ are positive constants satisfying $b_1=c_1\geq\frac{p_{max}}{4\lambda_1(Q)}$, $b_2\geq\frac{p_{max}}{\lambda_1(Q)}$, $c_2\geq\frac{p_{max}(\overline{d}+u_{max})}{\lambda_1(Q)}$
\begin{theorem}\label{theorem3}
  With the given dynamics (\ref{b7}), (\ref{b8}) and observer (\ref{b35}), under Assumption \ref{tree}, $\alpha_i$ and $\beta_i$ converges to $\tilde{x}_i$ and $\tilde{v}_i$ within a fixed-time $T_b=t_{b1}+t_{b_2}$.
\end{theorem}
\begin{proof}
Let $\tilde{x}_i=x_i-x_0$, $\tilde{v}_i=v_i-v_0$ and $\tilde{\alpha}_i=\alpha_i-\tilde{x}_i$, $\tilde{\beta}_i=\beta_i-\tilde{v}_i$. Then we have
\begin{align}\label{bb35}
\dot{\tilde{\alpha}}_i=&\tilde{\beta}_i-2b_1[h_2(t)+2]{\sum_{j=0}^{n}a_{ij}(\tilde{\alpha}_i-\tilde{\alpha_j})}\nonumber\\
&-b_2[h_2(t)+2]{\rm sgn}\bigg[\sum_{j=0}^{n}a_{ij}(\tilde{\alpha}_i-\tilde{\alpha_j})\bigg],\nonumber\\
\dot{\tilde{\beta}}_i=&-2c_1[h_2(t)+2]{\sum_{j=0}^{n}a_{ij}(\tilde{\beta}_i-\tilde{\beta_j})}\nonumber\\
&-c_2[h_2(t)+2]{\rm sgn}\bigg[\sum_{j=0}^{n}a_{ij}(\tilde{\beta}_i-\tilde{\beta_j})\bigg]-d_i+u_0.
\end{align}
Let $z_i=\sum_{j=0}^{n}a_{ij}(\tilde{\beta_i}-\tilde{\beta}_j)$ and then one obtains
\begin{align}\label{b36}
\dot{\tilde{\beta}}_i=-2c_1[h_2(t)+2]z_i-c_2[h_2(t)+2]{\rm sgn}(z_i)-d_i+u_0.
\end{align}
Differentiating $z_i$ against time one has
\begin{align}\label{b37}
  \dot{z}_i=&-2c_1[h_2(t)+2]\sum_{j=0}^{n}a_{ij}(z_i-z_j)\nonumber\\
  &-c_2[h_2(t)+2]\bigg\{\sum_{j=0}^{n}a_{ij}[{\rm sgn}(z_i)-{\rm sgn}(z_j)]\bigg\}\nonumber\\
  &-\sum_{j=0}^{n}a_{ij}(d_i-d_j)+a_{i0}u_0.
\end{align}

According to $z=H\tilde{\beta}$, where $H$ is a nonsingular matrix. Thus if $z$ converges to zero, $\tilde{\beta}$ converges as well.
Construct a Lyapunov candidate as
\begin{align}\label{b38}
  V_5=\sum_{i=1}^{n}p_i[c_1z_i^2+c_2|z_i|].
\end{align}
Then, one has
\begin{align*}\label{b39}
  \dot{V}_5=&\sum_{i=1}^{n}p_i[2c_1z_i+c_2{\rm sgn}(z_i)]\times\nonumber\\
  &\bigg\{-2c_1[h_2(t)+2]\sum_{j=0}^{n}a_{ij}(z_i-z_j)\nonumber\\
  &-c_2[h_2(t)+2]\sum_{j=0}^{n}a_{ij}[{\rm sgn}(z_i)-{\rm sgn}(z_j)]\nonumber\\
  &-\sum_{j=0}^{n}a_{ij}(d_i-d_j)+a_{i0}u_0\bigg\}\nonumber\\
  =&-[h_2(t)+2]\sum_{i=1}^{n}p_i[2c_1z_i+c_2{\rm sgn}(z_i)]\times\nonumber\\
  &\bigg\{2c_1\sum_{j=0}^{n}a_{ij}(z_i-z_j)+c_2\sum_{j=0}^{n}a_{ij}[{\rm sgn}(z_i)-{\rm sgn}(z_j)]\bigg\}\nonumber\\
  &-\sum_{i=1}^{n}p_i[2c_1z_i+c_2{\rm sgn}(z_i)]\bigg[\sum_{j=0}^{n}a_{ij}(d_i-d_j)-a_{i0}u_0\bigg]\nonumber\\
  =&-[h_2(t)+2][2c_1 z+c_2{\rm sgn}(z)]^TPH[2c_1 z+c_2{\rm sgn}(z)]\nonumber
\end{align*}
\begin{align}
     &-\sum_{i=1}^{n}p_i[2c_1z_i+c_2{\rm sgn}(z_i)]\bigg[\sum_{j=0}^{n}a_{ij}(d_i-d_j)-a_{i0}u_0\bigg]\nonumber\\
  =&-[h_2(t)+2][2c_1 z+c_2{\rm sgn}(z)]^TQ[2c_1 z+c_2{\rm sgn}(z)]\nonumber\\
  &-\sum_{i=1}^{n}p_i[2c_1z_i+c_2{\rm sgn}(z_i)]\bigg[\sum_{j=0}^{n}a_{ij}(d_i-d_j)-a_{i0}u_0\bigg]\nonumber\\
  \leq&-\lambda_1(Q)[h_2(t)+2][2c_1 z+c_2{\rm sgn}(z)]^T[2c_1 z+c_2{\rm sgn}(z)]\nonumber\\
 &-\sum_{i=1}^{n}p_i[2c_1z_i+c_2{\rm sgn}(z_i)]\bigg[\sum_{j=0}^{n}a_{ij}(d_i-d_j)-a_{i0}u_0\bigg]\nonumber\\
   =&-\lambda_1(Q)[h_2(t)+1]\bigg\{\sum_{i=1}^{n}[4c_1^2z_i^2+4c_1c_2|z_i|+c_2^2]\bigg\}\nonumber\\
  &-\lambda_1(Q)\bigg\{\sum_{i=1}^{n}[4c_1^2z_i^2+4c_1c_2|z_i|+c_2^2]\bigg\}\nonumber\\
 &-\sum_{i=1}^{n}p_i[2c_1z_i+c_2{\rm sgn}(z_i)]\bigg[\sum_{j=0}^{n}a_{ij}(d_i-d_j)-a_{i0}u_0\bigg]\nonumber\\
 \leq&-4c_1\lambda_1(Q)[h_2(t)+1]\bigg\{\sum_{i=1}^{n}[c_1z_i^2+c_2|z_i|]\bigg\}\nonumber\\
 &-c_2\lambda_1(Q)\bigg\{\sum_{i=1}^{n}[4c_1|z_i|+c_2]\bigg\}\nonumber\\
  &+p_{max}(\overline{d}+u_{max})\bigg\{\sum_{i=1}^{n}[2c_1|z_i|+c_2]\bigg\}\nonumber\\
\leq&-\frac{4c_1\lambda_1(Q)}{p_{max}}[h_2(t)+1]\bigg\{\sum_{i=1}^{n}p_{max}[c_1z_i^2+c_2|z_i|]\bigg\}\qquad\quad\nonumber\\
\leq&-h_2(t)V_5-V_5\nonumber\\
\leq&-h_2(t)V_5.
\end{align}
When $t\in[0,t_{b1})$, $h_2(t)=h_{b1}(t)$. Using differential equation (\ref{b1}) one has $\lim_{t\rightarrow t_{b1}}V_5\leq(\frac{\delta}{1+\delta})^kV_5(0)\approx 0$. When $t\geq t_{b1}$, $\dot{V}_5\leq-V_5$. Then $V_5$ will converge exponentially and $|z_i|<<1$ is ensured. Therefore, the observer can successfully estimate the velocity disagreements between the leader and the followers in $t_{b1}$.

In the following proof, let $w_i=\sum_{j=0}^{n}a_{ij}(\tilde{\alpha_i}-\tilde{\alpha}_j)$. Then substituting $w_i$ into (\ref{bb35}) one obtains
\begin{align}\label{bb36}
\dot{\tilde{\alpha}}_i=\tilde{\beta_i}-2b_1[h_2(t)+2]w_i-b_2[h_2(t)+2]{\rm sgn}(w_i).
\end{align}
Differentiating $w_i$ against time one has
\begin{align}\label{bb37}
  \dot{w}_i=&z_i-2b_1[h_2(t)+2]\sum_{j=0}^{n}a_{ij}(w_i-w_j)\nonumber\\
  &-b_2[h_2(t)+2]\bigg\{\sum_{j=0}^{n}a_{ij}[{\rm sgn}(w_i)-{\rm sgn}(w_j)]\bigg\}.
\end{align}
When $t\in[t_{b1},t_{b1}+t_{b2})$, due to $|z_i|<<1$, $z_i$ in (\ref{bb37}) can be seen as a bounded disturbance in (\ref{b37}). The following proof is the same as before and not restated. Until now, the proof of Theorem \ref{theorem3} has been finished.
\end{proof}
\subsection{Fixed-Time Distributed Averaging Tracking Observer}
\begin{assumption}\label{dat}
  The topology graph for the $n$ agents is undirected and connected. Each agent can only receive the information form one reference signal.
\end{assumption}
\begin{lemma} \label{lemma1}\cite{Ghapani2019}
If $L\in\mathbb{R}^{n\times n}$ is a Laplacian matrix of a connected undirected graph and $D\in\mathbb{R}^{n\times m}$ is its relative incidence matrix. Then for any vector $z\in\mathbb{R}^n$ one has
\begin{align}\label{b4}
z^TLD{\rm sgn}(D^Tz)\geq \lambda_2(L)z^TD{\rm sgn}(D^Tz).
\end{align}
\end{lemma}
The distributed observer is given as below.
\begin{align}\label{b40}
  \dot{\alpha}_i=&-b_1h_2(t)\sum_{j=1}^{n}a_{ij}(\alpha_i-\alpha_j)\nonumber\\
  &-b_2\sum_{j=1}^{n}a_{ij}{\rm sgn}(\alpha_i-\alpha_j)+\beta_i,\nonumber\\
  \dot{\beta}_i=&-c_1h_2(t)\sum_{j=1}^{n}a_{ij}(\beta_i-\beta_j)\nonumber\\
  &-c_2\sum_{j=1}^{n}a_{ij}{\rm sgn}(\beta_i-\beta_j)+a_i^r,
\end{align}
where $b_1$, $b_2$, $c_1$, $c_2$ are positive constants satisfying $b_1=c_1\geq\frac{1}{2\lambda_2(L)}$, $b_2\geq 1$, $c_2>2a_{max}$; Moreover the initial states satisfy $\sum_{i=1}^{n}\alpha_i(0)=\sum_{i=1}^{n}x_i(0)$ and $\sum_{i=1}^{n}\beta_i(0)=\sum_{i=1}^{n}v_i(0)$.

Note that $\sum_{i=1}^{n}\dot{\alpha}_i=\sum_{i=1}^{n}\beta_i$ and $\sum_{i=1}^{n}\dot{\beta}_i=\sum_{i=1}^{n}a_i^r$. Therefore under the given initial states, we have $\sum_{i=1}^{n}\beta_i=\sum_{i=1}^{n}v_i$ and $\sum_{i=1}^{n}\alpha_i=\sum_{i=1}^{n}x_i$ all the time. Following from this, if all the observers achieve consensus in a fixed time, the average value of the multiple reference signals is obtained successfully.
\begin{theorem}\label{theorem4}
With the given dynamics (\ref{b8}), (\ref{bb9}) and observer (\ref{b40}), under Assumption \ref{dat}, $\alpha_i$ and $\beta_i$ converges to $\overline{r}$ and $\overline{f}$ within a fixed-time $T_b=t_{b1}+t_{b_2}$.
\end{theorem}
\begin{proof}
Construct a Lyapunov candidate as $V_6=\frac{1}{2}\beta^TL\beta$ and (\ref{b40}) can be written in the vector form as
\begin{align}\label{bb40}
  \dot{\alpha}=&-b_1h_2(t)L\alpha-b_2D{\rm sgn}(D^T\alpha)\nonumber\\
  \dot{\beta}=&-c_1h_2(t)L\beta-c_2D{\rm sgn}(D^T\beta)+a.
\end{align}
Then we have
\begin{align}\label{b41}
\dot{V}_6=&\beta^TL\dot{\beta}\nonumber\\
=&\beta^TL[-c_1h_2(t)L\beta-c_2D{\rm sgn}(D^T\beta)+a]\nonumber\\
\leq&-c_1h_2(t)\beta^TLL\beta-c_2\beta^TLD{\rm sgn}(D^T\beta)+\beta^TLa\nonumber\\
\leq&-c_1h_2(t)\lambda_2(L)\beta^TL\beta-c_2\beta^TD{\rm sgn}(D^T\beta)+\beta^TLa\nonumber\\
\leq&-h_2(t)V_6-c_2||D^T\beta||_1+(D^T\beta)^TD^Ta\nonumber\\
\leq&-h_2(t)V_6-(c_2-2a_{max})||D^T\beta||_1\nonumber\\
\leq&-h_2(t)V_6.
\end{align}
When $t\in[0,t_{b1})$, $h_2(t)=h_{b1}(t)$. Using differential equation (\ref{b1}) one has $\lim_{t\rightarrow t_{b1}}V_6\leq(\frac{\delta}{1+\delta})^kV_6(0)\approx 0$. When $t\geq t_{b1}$, $\dot{V}_6\leq-(c_2-2a_{max})||D^T\beta||_1$. Then compared with (\ref{b32}) one has $V_6$ will converge in finite time after $t_{b1}$ and $\beta_i<<1$ is ensured.
Construct the Lyapunov candidate as $V_6=\frac{1}{2}\alpha^TL\alpha$
\begin{align}\label{b42}
\dot{V}_7=&\alpha^T L\dot{\alpha}\nonumber\\
=&\alpha^T L[-b_1h_2(t)L\alpha-c_2D{\rm sgn}(D^T\alpha)+\beta]\nonumber\\
=&-b_1h_2(t)\alpha^TLL\alpha-c_2\alpha^T LD{\rm sgn}(D^T\alpha)+\alpha L\beta]\nonumber\\
\leq&-b_1\lambda_2(L)h_2(t)\alpha^TL\alpha-c_2||D^T\alpha||_1+(D^T\alpha)^TD^T\beta\nonumber\\
\leq&-h_2(t)V.
\end{align}
Due to $\beta_i<<1$, $|\beta_i-\beta_j|<<1$ is ensured. Then one has $c_2||D^T\alpha||_1>(D^T\alpha)^TD^T\beta$. The following is the same as before and hence omitted. Until now, the proof of Theorem \ref{theorem4} has been finished.
\end{proof}
\begin{remark}
All the observers proposed in this paper can be extended to that of high-order multi-agent systems by using more time-based
 generators and more integrators.
\end{remark}
\subsection{Distributed Consensus Tracking and Distributed Average Tracking Control}
After designing the observers, the next step is to design the controllers by using the information provided by the observers. In this subsection, the control inputs for the distributed consensus tracking and the distributed average tracking will be given.
\begin{theorem}\label{tracking}
  Under dynamics (\ref{b7}), (\ref{b8}) and Assumption \ref{assumption1} (Assumption \ref{tree}), with the observer (\ref{b26}) (observer (\ref{b35})) and the control input
  \begin{align}\label{bt1}
     u_i=
  \begin{cases}
  =0,t\in[0,T_c),\\
=-\frac{1}{2}\dot{h}_1(t)\alpha_i-(\frac{1}{2}h_1(t)+1)\beta_i\\
-\frac{1}{2}h_1(t)s_i-\rho{\rm sgn}(s_i),t\geq T_b,
  \end{cases}
  \end{align}
  where $T_c\geq T_b$, $s_i=(\frac{1}{2}h_1(t)+1)\alpha_i+\beta_i$, $\rho\geq d_{max}+u_{max}+1$, the fixed-time distributed consensus tracking for double-integrator-type multi-agent systems is solved. Further more, the upper-bounded convergence time is $T_a+T_c$.
\end{theorem}
\begin{proof}
 When $t\geq T_{b}$, one has $\alpha_i=x_i-x_0$ and $\beta_i=v_i-v_0$. Then we have
 \begin{align}\label{bt11}
   \dot{s}_i=&\frac{1}{2}\dot{h}_1(t)\alpha_i+(\frac{1}{2}h_1(t)+1)\dot{\alpha}_i+\dot{\beta}_i\nonumber\\
   =&\frac{1}{2}\dot{h}_1(t)\alpha_i+(\frac{1}{2}h_1(t)+1)\beta_i+\dot{v}_i-\dot{v}_0\nonumber\\
   =&\frac{1}{2}\dot{h}_1(t)\alpha_i+(\frac{1}{2}h_1(t)+1)\beta_i+u_i+d_i-u_0.
 \end{align}
 Then substitute (\ref{bt1}) into (\ref{bt11}), one has
 \begin{align}\label{bt12}
    \dot{s}_i=-\frac{1}{2}h_1(t)s_i-\rho{\rm sgn}(s_i)+d_i-u_0.
 \end{align}
 The other part of the proof is the same as Theorem \ref{t1} and omitted. Thus the whole proof has been finished.
\end{proof}
%\begin{remark}
%    Although there is a little error for $\beta_i=v_i-v_0$ at the near moment of $T_b$, it has no influence on the final results. The error can be diminished by choosing the appropriate $\rho$. Moreover, the error converges to zero very quickly.
%\end{remark}
\begin{theorem}\label{dtracking}
  Under dynamics (\ref{b8}), (\ref{bb9}) and Assumption \ref{dat}, with the observer (\ref{b40}) and the control input
  \begin{align}\label{bt2}
       u_i=
  \begin{cases}
  =0,\in[0,T_c),\\
=-\frac{1}{2}\dot{h}_1(t)(x_i-\alpha_i)-(\frac{1}{2}h_1(t)+1)(v_i-\beta_i)\\
-\frac{1}{2}h_1(t)s_i-\rho{\rm sgn}(s_i),t\geq T_b,
  \end{cases}
  \end{align}
  where $s_i=(\frac{1}{2}h_1(t)+1)(x_i-\alpha_i)+(v_i-\beta_i)$, $\rho\geq d_{max}+a_{max}+1$, the fixed-time distributed average tracking problems for double-integrator-type multi-agent systems is solved. Further more, the upper-bounded convergence time is $T_a+T_c$.
\end{theorem}
\begin{proof}
  When $t\geq T_{b}$, one has $\alpha_i=\overline{r}$ and $\beta_i=\overline{f}$. Then we have
 \begin{align}\label{bt111}
   \dot{s}_i=&\frac{1}{2}\dot{h}_1(t)(x_i-\alpha_i)+(\frac{1}{2}h_1(t)+1)(v_i-\beta_i)\nonumber\\
  &+u_i+d_i-\dot{\beta}_i.
 \end{align}
 Then substitute (\ref{bt2}) into (\ref{bt111}), one has
 \begin{align}\label{bt122}
    \dot{s}_i=-\frac{1}{2}h_1(t)s_i-\rho{\rm sgn}(s_i)+d_i-\overline{a}.
 \end{align}
 The other part of the proof is the same as Theorem \ref{t1} and omitted. Thus the whole proof has been completed.
\end{proof}
%%%%%%%%%%%%%%%%%%%%%%%%%%%%%     Numerical simulations      %%%%%%%%%%%%%%%%%%%%%%%%%%%%%5%%
%to manifest the effectiveness of the result in Theorem \ref{theorem1} and \ref{theorem2}.
\section{Numerical simulations}\label{sec4}

\textbf{Example 1. }\label{example1}
A simulation for Theorem \ref{t1} is given as follows. Set $k=2$, $\delta=0.01$, $\varrho={\rm sin}(t)$, $\rho=2$, $t_{a1}=t_{a2}=3$, $z_1(0)=200$ and $z_2(0)=100$. The results with upper-bounded convergence time $T_a=6s$ are shown in Fig. 1.
\begin{figure}[htbp]
	\center{
		\includegraphics[width=6.6cm]{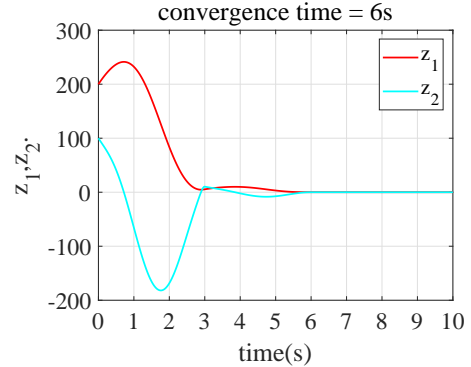}\\
		\caption{The results of the sliding mode control in Example 1.} }
\end{figure}
%$r_i(t)=sin(t)+20i$. The upper bounded convergence time is set as $t_f=4$ and $t_f=1$. The fixed-time practical DAT is achieved in Fig. 1.
%Under Assumptions \ref{assumption1} and \ref{assumption2}, consider a multi-agent system with six agents described by (\ref{a1}) and the topology graph is shown in Fig. 1. Set the reference signals with bounded derivatives as $r_i(t)=sin(t)+20i$ and the initial states of agents as $x(0)=[20,40,60,80,100,120]^T$. In order to limit the disagreement to a desired level, we set $\delta=0.001$.
%Implementing protocols (\ref{a2}) and (\ref{au}) with $t_f=4$ and $t_f=1$, the fixed-time DAT is achieved in Fig. 2.

\textbf{Example 2. }\label{example2}
A simulation for Theorem \ref{tracking} under Assumption \ref{assumption1} is given as follows. Consider a multi-agent system described by (\ref{b7}) and (\ref{b8}) with the undirected communication topology in Fig. 2. Set $k=2$, $\delta=0.01$, $\rho=8$, $b_1=c_1=4$, $b_2=1$, $c_2=8$, $t_{a1}=t_{a2}=3$, $t_{b1}=t_{b2}=1.5$, $T_b=T_c$, $u_0=1+5{\rm sin}(t)$ and $x_i(0)$, $v_i(0)$ with random states. The results with upper-bounded convergence time $T_a+T_c=9s$ are shown in Fig. 3 and Fig. 4.
\begin{figure}[htbp]
	\center{
		\includegraphics[width=2.8cm]{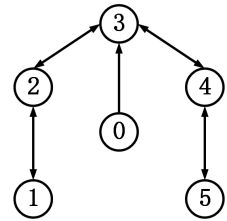}\\
		\caption{The communication topology in Example 2.} }
\end{figure}
\begin{figure}[htbp]
	\center{
		\includegraphics[width=6.6cm]{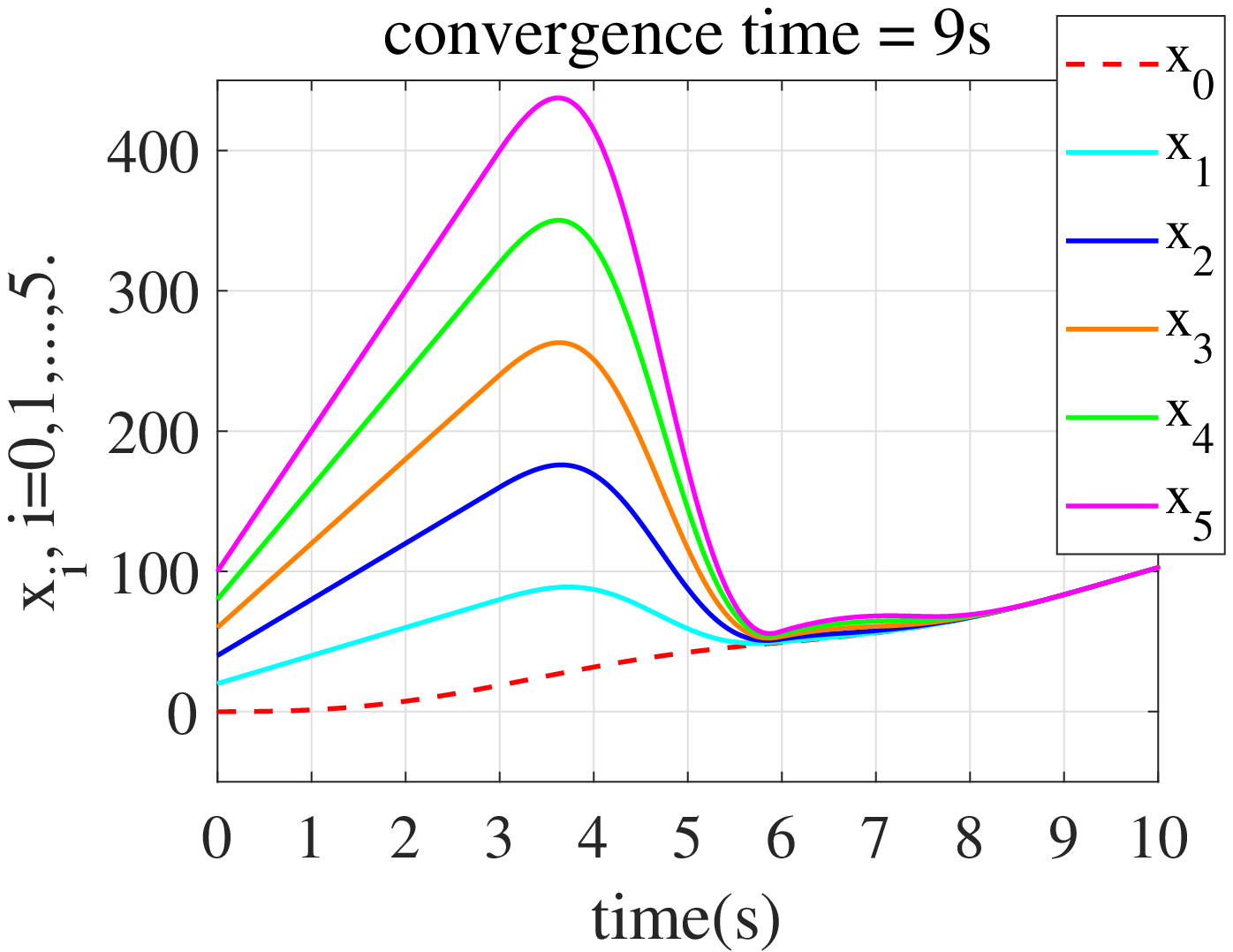}\\
		\caption{The positions of the agents in Example 2.} }
\end{figure}
\begin{figure}[htbp]
	\center{
		\includegraphics[width=6.6cm]{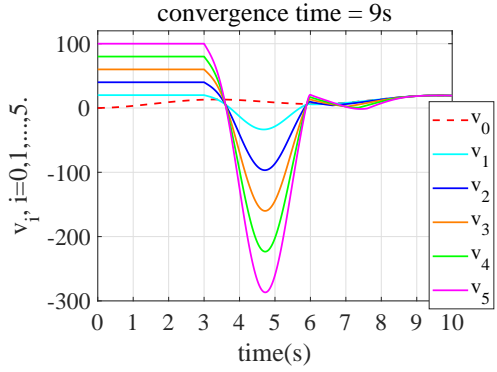}\\
		\caption{The velocities of the agents in Example 2.} }
\end{figure}

{\textbf{Example 3. }\label{example3}
A simulation for Theorem \ref{tracking} under Assumption \ref{tree} is given as follows. Consider a multi-agent system described by (\ref{b7}) and (\ref{b8}) with the directed communication topology in Fig. 5. Set $k=2$, $\delta=0.01$, $\rho=21$, $b_1=c_1=2$, $b_2=7$, $c_2=34$, $t_{a1}=t_{a2}=2$, $t_{b1}=t_{b2}=1$, $T_b=T_c$,  $u_0=2+18{\rm sin}(t)$ and $x_i(0)$, $v_i(0)$ with random states. The results with upper-bounded convergence time $T_a+T_c=6s$ are shown in Fig. 6 and Fig. 7.
\begin{figure}[htbp]
	\center{
		\includegraphics[width=3cm]{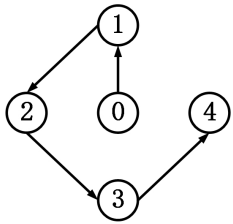}\\
		\caption{The communication topology in Example 3.} }
\end{figure}
\begin{figure}[htbp]
	\center{
		\includegraphics[width=6.6cm]{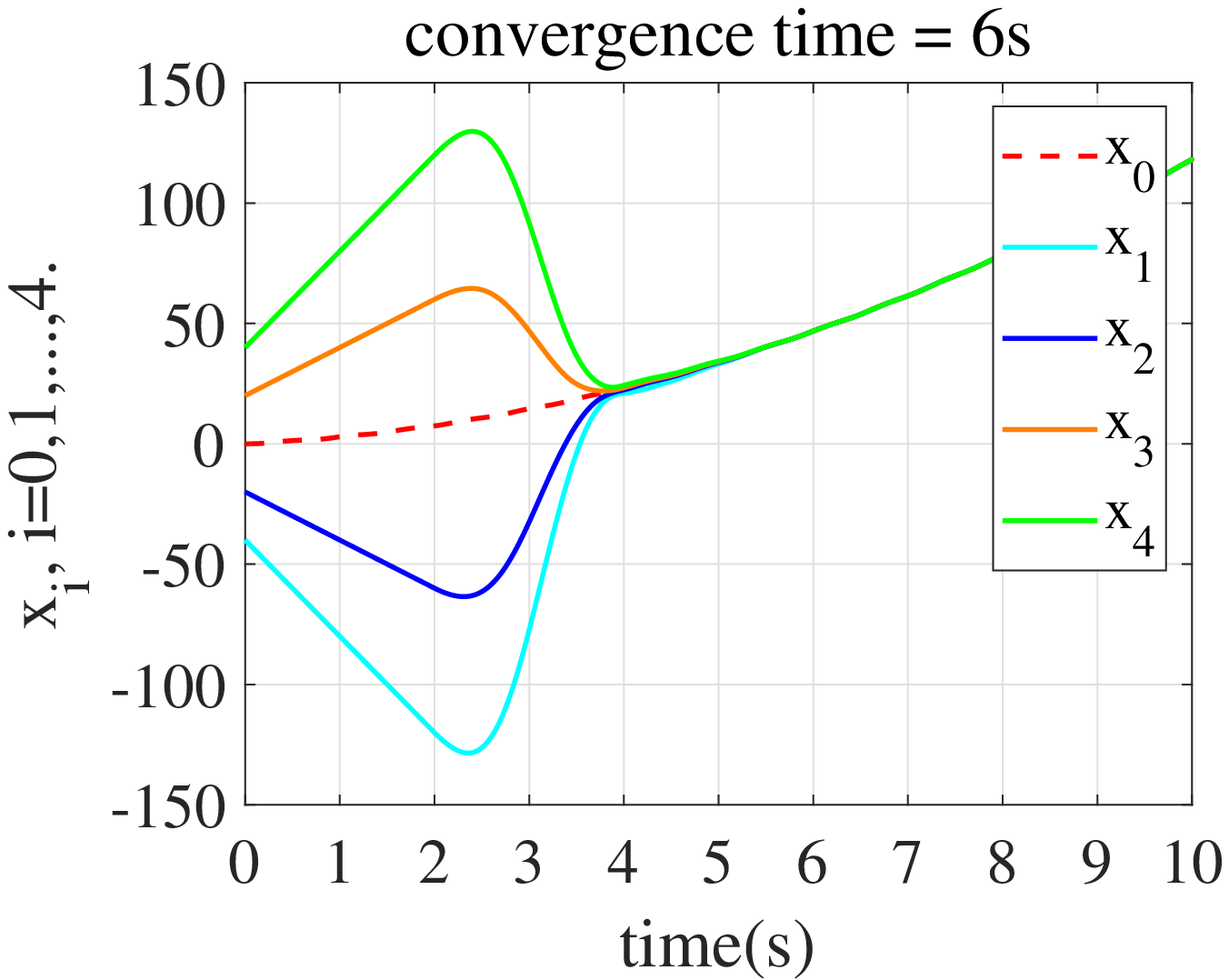}\\
		\caption{The positions of the agents in Example 3.} }
\end{figure}
\begin{figure}[htbp]
	\center{
		\includegraphics[width=6.6cm]{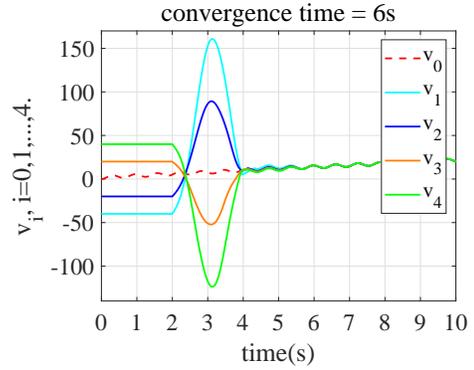}\\
		\caption{The velocities of the agents in Example 3.} }
\end{figure}

{\textbf{Example 4. }\label{example4}
A simulation for Theorem \ref{dtracking} under Assumption \ref{dat} is given as follows. Consider a multi-agent system described by (\ref{b8}) and (\ref{bb9}) with the undirected communication topology in Fig. 8. Set $k=2$, $\delta=0.01$, $\rho=63$, $b_1=c_1=0.25$, $b_2=1$, $c_2=123$, $t_{a1}=t_{a2}=4$, $t_{b1}=t_{b2}=2$, $T_b=T_c$, $a_1^r=41+20{\rm sin}(5t)$, $a_2^r=51+10{\rm sin}(5t)$, $a_3^r=30+30{\rm sin}(5t)$, $a_4^r=40+20{\rm sin}(5t)$ and $x_i(0)$, $v_i(0)$ with random states. The results with upper-bounded convergence time $T_a+T_c=12s$ are shown in Fig. 9 and Fig. 10.
\begin{figure}[htbp]
	\center{
		\includegraphics[width=3cm]{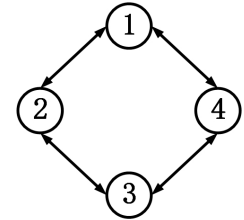}\\
		\caption{The communication topology in Example 4.} }
\end{figure}
\begin{figure}[htbp]
	\center{
		\includegraphics[width=6.6cm]{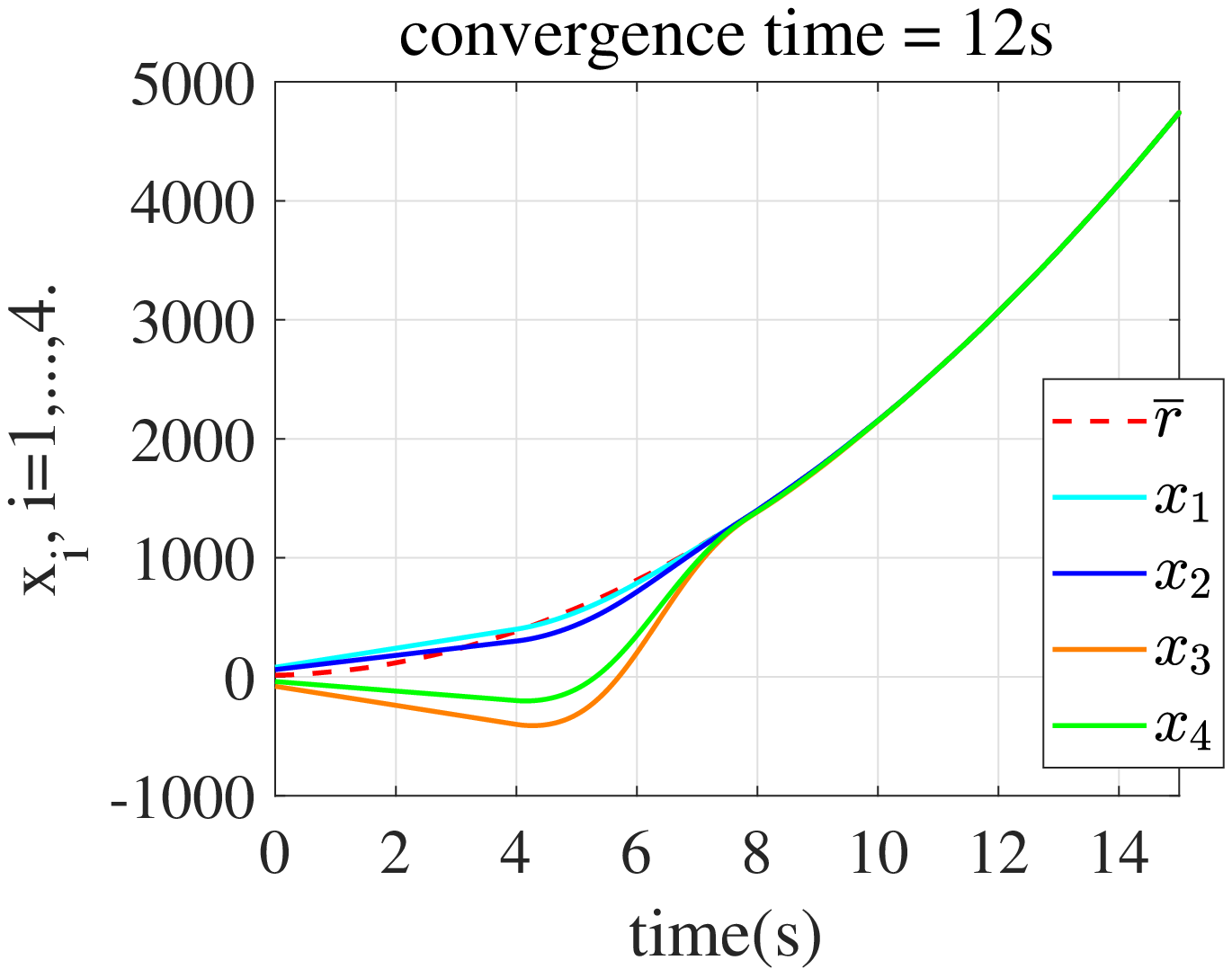}\\
		\caption{The positions of the agents in Example 4.} }
\end{figure}
\begin{figure}[htbp]
	\center{
		\includegraphics[width=6.6cm]{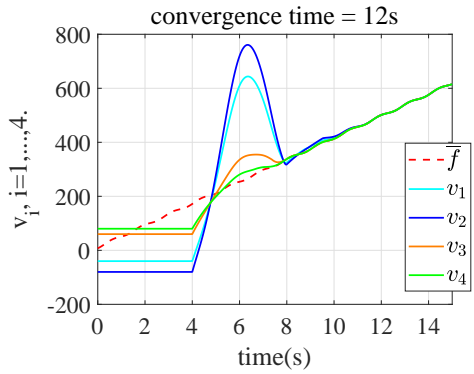}\\
		\caption{The velocities of the agents in Example 4.} }
\end{figure}
\section{Conclusions}\label{sec5}
In this paper, both the fixed-time distributed consensus tracking and the fixed-time distributed average tracking problems for double-integrator-type multi-agent systems are solved by using time-based generators. Different from traditional fixed-time methods, the time-based generator approach can directly predesign the fixed time, which is of great significance in reality. But the tradeoff is the introduce of time dependent function. Moreover, it is trivial to extend the fixed-time sliding mode control method in this article to Euler-Lagrange systems. By combining the fixed-time sliding mode control method of Euler-Lagrange systems and the observers in this article, the fixed-time distributed consensus tracking and distributed average tracking for multiple Euler-Lagrange systems can be achieved. Also, the fixed-time distributed consensus tracking problem for single-integrator multi-agent systems under directed graph can be solved by devising a controller similar to the velocity observer in (\ref{b40}).
\bibliographystyle{IEEEtran}

\end{document}